\documentclass[12pt]{amsart}

\usepackage{amsfonts, amsmath, amssymb}
\usepackage{graphicx,color}
\usepackage{listings}
\usepackage{algorithm}
\usepackage{algpseudocode}
\usepackage{cancel} 
\usepackage{wrapfig}
\usepackage{color}
\usepackage{tikz}
\usepackage{url}

\newtheorem{theorem}{Theorem}[section]

\theoremstyle{definition}
\newtheorem{definition}[theorem]{Definition}

\newcommand{\LLL}{\mathfrak{L}}

\newcommand{\RRR}{\mathfrak{R}}
\newcommand{\ARRR}{\mathfrak{R}^t}

\newcommand{\ta}{\mathrm{TA}}

\newcommand{\nta}{\mathrm{NTA}}

\newcommand{\ntaeps}{\mathrm{eNTA}}

\newcommand{\discta}[1]{{#1}_d}

\newcommand{\Actions}{\Sigma}

\newcommand{\Locs}{\mathcal{Q}}
\newcommand{\Acc}{\mathcal{F}}
\newcommand{\Clocks}{\mathcal{C}}
\newcommand{\ResetClocks}{\mathcal{C}_{rst}}

\newcommand{\clk}{c}

\newcommand{\Trans}{\mathcal{T}}
\newcommand{\Guards}{\mathcal{G}}

\newcommand{\Edges}{\mathcal{E}}
\newcommand{\Delays}{\mathcal{D}}
\newcommand{\Vertices}{\mathcal{V}}
\newcommand{\Weights}{\mathcal{W}}

\newcommand{\ClockVal}{\mathcal{V}}
\newcommand{\zerov}{\textbf{0}}

\newcommand{\Naturals}{\mathbb{N}}
\newcommand{\ZNaturals}{\mathbb{N}_0}

\newcommand{\Reals}{\mathbb{R}}
\newcommand{\PReals}{\mathbb{R}_{\geq 0}}

\newcommand{\frc}[1]{\left\{{#1}\right\}}

\newcommand{\powerset}[1]{\mathcal P \left({#1}\right) }

\usepackage{color}

\begin{document}

\title{On the Distance between Timed Automata}
\author{Amnon Rosenmann}
\address{Graz University of Technology, Steyrergasse 30, A-8010 Graz, Austria}
\email{rosenmann@math.tugraz.at}

\date{}
\maketitle

\begin{abstract}
The problem of inclusion of the language accepted by timed automaton $A$ (e.g., the implementation) in the language accepted by $B$ (e.g., the specification) is, in general, undecidable in the class of non-deterministic timed automata.
In order to tackle this disturbing problem we show how to effectively construct deterministic timed automata $\discta{A}$ and $\discta{B}$ that are discretizations (digitizations) of the non-deterministic timed automata $A$ and $B$ and differ from the original automata by at most $\frac{1}{6}$ time units on each occurrence of an event.
Language inclusion in the discretized timed automata is decidable and it is also decidable when instead of $\LLL(B)$ we consider $\overline{\LLL(B)}$, the closure of ${\LLL(B)}$ in the Euclidean topology:
if $\LLL(\discta{A}) \nsubseteq \LLL(\discta{B})$ then $\LLL(A) \nsubseteq \LLL(B)$ and if 
$\LLL(\discta{A}) \subseteq \LLL(\discta{B})$ then $\LLL(A) \subseteq \overline{\LLL(B)}$. 

Moreover, if $\LLL(\discta{A}) \nsubseteq \LLL(\discta{B})$ we would like to know how far away is $\LLL(\discta{A})$ from being included in $\LLL(\discta{B})$. For that matter we define the distance between the languages of timed automata as the limit on how far away a timed trace of one timed automaton can be from the closest timed trace of the other timed automaton.
We then show how one can decide under some restriction whether the distance between two timed automata is finite or infinite.

\end{abstract}

\section{Introduction}
\label{sec:intro}
Timed automaton (TA) was introduced by Alur and Dill \cite{ta} as an abstract model for real-time systems by extending finite automaton with continuous clocks.
When the $\ta$s are non-deterministic then a fundamental problem of language inclusion is, in general, undecidable, for example, whether the set of timed traces of the $\ta$ representing the implementation is included in that of the specification.
This lead to imposing restrictions on and modifications to non-deterministic TAs in order to achieve decidability (see \cite{bbbb,ta-eps,era,updatable-ta,Ouak_one_clk,Ouak_time_bound,bounded-time,Dec_TA_Survey,bound-det} for a partial list).
Another approach  was to allow robustness in the language \cite{GHJ97} or perturbations in the clocks \cite{purturbed-ta} (see also \cite{BMS15}).
The problem is that by allowing a fixed imprecision, undecidability problems due to working over continuous time do not vanish.

Digitization of timed systems, where basic decision problems like language inclusion are decidable, was considered, for example, in \cite{HMP92,O02,OW03}.
But, as stated in \cite{OW03}, the implementation should be 'closed under digitization' and
the specification should be 'closed under inverse digitization' in order to be able to reduce the language inclusion problem from the continuous world to the discretized one.
In \cite{irta} the authors construct TAs with reset only on integral time and demonstrate the decidability of the language inclusion problem $\LLL(A) \subseteq \LLL(B)$ in case $B$ is an integer reset TA.
In this paper we go further with this approach.
The idea is to work in the setting of discretized time, but without restricting or modifying the definition of a TA. The discretization is over intervals which are smaller than 1 time unit so that although we work in the discretized setting we are able to check for exact occurrence of events also outside integral time. 
For this matter we construct discretized TAs that enable effective comparison of the languages of the original TAs. 
The discretized TA stays within a distance of $\frac{1}{6}$ time units from the original TA (the distance can, in fact, be as small as we like, in the cost of complexity, but that won't improve our knowledge about the inclusion of the languages of the original automata), a goal that is achieved through the introduction of an additional clock, $t$, that measures absolute time.
Now, instead of comparing directly the language of two TAs, a problem which is in general undecidable, we can compare their discretized TAs and have the following (see Theorem~\ref{th:inclusion_gap}): if $\LLL(\discta{A})$, the language of the discretized TA of $A$, is not included in $\LLL(\discta{B})$, the language of the discretized TA of $B$, then the same holds for $\LLL(A)$ with respect to $\LLL(B)$.
If, however, $\LLL(\discta{A}) \subseteq \LLL(\discta{B})$ then $\LLL(A)$ is included in the topological closure of $\LLL(B)$. 

The next natural question, in case $\LLL(A) \nsubseteq \LLL(B)$, is how far away is a timed trace of $\LLL(A)$ from all timed traces of $\LLL(B)$, that is, what is the conformance distance $c(\LLL(A), \LLL(B))$, the distance of $\LLL(A)$ from being conformed with $\LLL(B))$.
When an untimed word of $\LLL(A)$ is not in  $\LLL(B)$ or when a transition in $A$ which is not bound in time is not met with a similar transition in $\LLL(B)$ of the same action label then $c(\LLL(A), \LLL(B)) = \infty$ and the existence of these cases is decidable.
A more challenging question is whether there is a sequence of timed traces of $\LLL(A)$ which tend to diverge from $\LLL(B)$, causing $c(\LLL(A), \LLL(B)) = \infty$.
For example, it may happen that due to imprecisions or delays in a real system, a TA model is changed to allow wider time intervals around actions compared to the more idealistic previous model.
It is then necessary to check whether or not this extended freedom is controlled and the distance between the two TAs stays within a reasonable bound (see \cite{BMS15} regarding an ideal model versus a realistic model).
Moreover, an algorithm based on the approach suggested here may find the timed traces that deviate from the allowed distance between two timed languages.
Further applications for computing the distance may be when safety properties include time restrictions for specific set of timed traces, given as timed automata, and we want to check these timed traces with respect to the implementation model.
In general, in a design of a computerized system, e.g. a network, that contains timing changes, a relaxed equivalence verification may allow bounded perturbations in time that needed to be checked. 

Computing the distance between TAs (or their languages), even between discretized TAs, may be quite complex.
Here we concentrate on the problem of deciding whether the distance is infinite.
It is not clear to us whether this problem is decidable in general, but for a (perhaps) restricted version of it we construct an algorithm that solves it.

\section{Timed Automaton}
\label{sec:ta}
A timed automaton is an abstract model of temporal behavior of real-time systems.
It is a finite automaton with \emph{locations} and \emph{transitions} between them, extended with a finite set of (continuous) clocks defined over $\PReals$.
A transition at time $t$ can occur only if the condition expressed as a \emph{transition guard} is satisfied at $t$.
A transition guard is a conjunction of constraints of the form $\clk \sim n$, where $\clk$ is a clock, $\sim \ \in \{<,\leq, =,\geq, >\}$ and $n \in \ZNaturals = \Naturals \cup \{0\}$.
Each transition is labeled by some \emph{action} $a \in \Actions$ and some of the clocks may be reset to zero. 
In $\nta$, the class of non-deterministic timed automata, and unlike deterministic TAs, it may occur that two transitions from the same location $q$ can be taken at the same time and with the same action but to two different locations $q'$ and $q''$.
\begin{definition}[Timed automaton]
\label{def:ntaeps}
A \emph{non-deterministic timed automaton} $A \in \nta$ is a tuple $(\Locs, q_0, \Acc, \Actions, \Clocks,
\Trans)$, where:
\begin{enumerate}
\item $\Locs$ is a finite set of locations and $q_0 \in \Locs$ is the initial location;
\item $\Acc \subseteq \Locs$ is the set of accepting locations;
\item $\Actions$ is a finite set of transition labels, called actions;
\item $\Clocks$ is a finite set of clock variables;
\item $\Trans \subseteq \Locs \times \Actions \times \Guards \times \powerset{\Clocks} \times \Locs$ is a finite set of transitions of the form $(q, a, g, \ResetClocks, q')$, where:
\begin{enumerate}
\item $q,q' \in \Locs$ are the source and the target locations respectively;
\item $a \in \Actions$ is the transition action;
\item $g \in  \Guards$ is the \emph{transition guard};
\item $\ResetClocks \subseteq \Clocks$ is the subset of clocks to be reset.
\end{enumerate}
\end{enumerate}
\end{definition}
A clock \emph{valuation} $v (\clk)$ is a function $v:\Clocks \to \PReals$. 
We denote by $\ClockVal$ the set of all clock valuations and by
$\textbf{d}$ the valuation which assigns the value $d$ to every clock.
Given a valuation $v$ and $d \in \PReals$, we define $v+d$ to be the valuation $(v+d)(\clk) := v(\clk)+d$ for every $ \clk \in \Clocks$.
The valuation $v[\ResetClocks]$, $\ResetClocks \subseteq \Clocks$, is defined to be $v[\ResetClocks](c) = 0$ for $c \in \ResetClocks$ and $v[\ResetClocks](c) = v(c)$ for $c \notin \ResetClocks$.

The \emph{semantics} of $A \in \nta$ is given by the \emph{timed transition system}
$[[A]] = (S, s_0, \PReals, \Actions, T)$, where:
\begin{enumerate}
\item $S = \{(q,v) \in \Locs \times \ClockVal \}$ is the set of states, with
$s_0 = (q_0, \zerov)$ the initial state;
\item $T \subseteq S \times (\Actions \cup \PReals) \times S$ is the transition relation.
The set $T$ consists of
\begin{enumerate}
\item \emph{Timed transitions (delays):} $(q,v) \xrightarrow{d} (q, v+d)$, where $d \in \PReals$;
\item \emph{Discrete transitions (jumps):} $(q,v) \xrightarrow{a} (q',v')$, where $a \in \Actions$ and there exists a transition $(q, a, g, \ResetClocks, q')$ in $\Trans$, such that the valuation $v$ satisfies the guard $g$ and $v' = v[\ResetClocks]$.
\end{enumerate}
\end{enumerate}

A (finite) \emph{run} $\varrho$ on $A \in \ntaeps$ is a sequence of alternating timed and discrete transitions
of the form 
$$(q_{0}, \zerov) \xrightarrow{d_{1}} (q_{0}, \textbf{d}_1) \xrightarrow{a_{1}} (q_{1}, v_{1}) \xrightarrow{d_{2}} \cdots 
\xrightarrow{d_{k}} (q_{k-1}, v_{k-1} + d_{k}) \xrightarrow{a_{k}} (q_{k}, v_{k}).
$$
The run $\varrho$ on $A$ induces the \emph{timed trace} (\emph{timed word})
$$
\tau = (t_{1}, a_{1}), (t_{2}, a_{2}), \ldots, (t_{k}, a_{k}),
$$
with $a_i \in \Actions$ and $t_{i} = \Sigma_{j=1}^{i} d_i$. 
The language $\LLL(A)$ consists of the set of timed traces that are obtained from the runs that end in accepting locations.
We remark that for simplification of presentation we did not include the location invariants in the definition of timed automata since they are more of a 'syntactic sugar': the invariants of location $q$ are composed of upper bounds to the values of the clocks while being in $q$, but these constraints can be incorporated in the transition guards to $q$ (for the clocks that are not reset at the transitions) and in those transitions that emerge from $q$, thus not affecting $\LLL(A)$.
\section{Augmented Region Automaton}
\label{sec:ARA}
Given a (finite) timed automaton $A$, the region automaton $\RRR(A)$ \cite{ta} is a finite \emph{discretized} version of $A$, such that time is abstracted and both automata define the same untimed language.
Instead of looking at the clocks-space as a continuous space it is partitioned into regions.
Suppose that the maximal integer appearing in the transition guards of $A$ is $M$, then we denote by $\top$ a value of a clock which is greater than $M$.
The regions partition the space of clock valuations into equivalent classes, where two valuations belong to the same equivalent class if and only if they agree on the clocks with $\top$ value and on the integral parts and the order among the fractional parts of the other clocks.
The edges of $\RRR(A)$ are labeled by the transition actions and they correspond to the actual transitions that occur in the runs on $A$.

The \emph{augmented region automaton}, denoted $\ARRR(A)$, is defined as in \cite{R19}.
First, we add to $A$ a clock $t$ that measures absolute time, is never reset to $0$ and does not affect the runs and timed traced of $A$.
Secondly, we want to construct $\ARRR(A)$ in a way that keeps track of absolute time and regain much of the information that is lost when passing from the timed automaton $A$ to the regular region automaton $\RRR(A)$. 
But since $t$ does not appear in the transition guards of $A$, we need not know the exact value of the integral part of $t$ but just how much time passes between two consecutive transitions.
Thus, we assign $t$ in $\ARRR(A)$ only two time-regions: $\{0\}$ and $(0,1)$.
However, in order to keep track of the absolute time that passes, each edge is assigned a 'weight', the time difference in the integral part of $t$ between the target and the source regions.
The ordering among the fractional part of the clocks does, however, take that of $t$ into account. 
Overall, the number of regions of $\ARRR(A)$ is clearly finite (although potentially exponentially large).
\begin{definition}
	\label{def:aug_region_automaton}
	Given a non-deterministic timed automaton $A$ with clocks $x_1, \ldots, x_s$ extended with absolute-time clock $t$, a corresponding augmented region automaton $\ARRR(A)$ is a tuple $(\Vertices, v_0, \Edges, \Actions, \Weights^*)$, where:
	\begin{enumerate}
		\item $\Vertices$ is the set of vertices.
		Each vertex is a triple $(q, {\bf n}, \Delta)$,	where $q$ is a location of $A$ and $r=({\bf n}, \Delta)$ is a region, with ${\bf n} = (n_1, \ldots, n_s) \in \{0, 1, \ldots, M, \top \}^{s}$
		consisting of the integral parts of the clocks $x_1, \ldots, x_s$
		and $\Delta$ is the simplex (hyper-triangle) with vertices in the lattice $\ZNaturals^{s+1}$ of all points that satisfy a fixed ordering of the fractional parts of the clocks $t=x_0, x_1, \ldots, x_s$:
		\begin{equation}
		\label{eq:simplex}
		0 \preceq_0 \frc{x_{i_0}} \preceq_1 \frc{x_{i_1}} \preceq_2 \cdots \preceq_s \frc{x_{i_s}} < 1,
		\end{equation} 
		where $\preceq_i \, \in \{=, < \}$.
		\item $v_0 = (q_0, {\bf 0}, {\bf 0})$ is the initial vertex, where $q_0$ is the initial location of $A$ and $({\bf 0}, {\bf 0})$ indicates that all clocks have value $0$.
		\item $\Edges$ is the set of edges.
		There is an edge $(q, r) \xrightarrow{a} (q',r')$ if and only if there is a run on $A$ containing $(q, v) \xrightarrow{d} (q, v + d) \xrightarrow{a} (q', v')$, such that, the clock valuation $v$ is in the region $r$ and $v'$ is in $r'$.
		\item $\Actions$ is the finite set of actions.
		\item $\Weights^*$ is the set of weights $m$ on the edges calculated as
		 $m = \lfloor t_1 \rfloor - \lfloor t_0 \rfloor \in [0..M]$, where $\lfloor t_1 \rfloor$ ($\lfloor t_0 \rfloor$) is the integral part of $t$ in the target (source) location in a corresponding run on $A$.
		 There may be more than one edge between two vertices of $\ARRR(A)$, each one with a distinguished weight.
		 A weight $m$ may be marked as $m^*$, representing infinitely-many consecutive values $m$, $m+1$, $m+2, \ldots$ as weights between the same two vertices, for example when the regular clocks passed the maximal value $M$.
	\end{enumerate}
\end{definition}
An augmented region automaton can be seen in Fig.~\ref{fig:disc_ta}(b) (the example is taken from \cite{ta}).

\section{Discretized Timed Automaton}
\label{sec:Approx}
After constructing the augmented region automaton $\ARRR(A)$, we turn it into a deterministic timed automaton $\discta{A}$ which discretizes (digitizes) $A$.
\begin{definition}
	\label{def:appta}
	A \emph{discretized timed automaton} $\discta{A}$ is a timed automaton constructed from the augmented region automaton $\ARRR(A)$ in the following way.
	\begin{enumerate}
		\item The directed graph structure of locations and edges of $\discta{A}$ is the same as that of $\ARRR(A)$.
		\item The transition labels (actions) are also as in $\ARRR(A)$.
		\item There is a single clock in $\discta{A}$, namely $t$, which is reset on each transition.
		\item The transition guards of $\discta{A}$ are of the following form.
		Let $e = v_0 \to v_1$ be an edge of $\ARRR(A)$, let $w(e)$
		be its weight and let $\{t_0\}, \{t_1\} \in [0,1)$ be any fractional parts of $t$ in the source and target regions.
		Let
		$$\delta = \frac{1}{2}(\lceil \{t_1\} \rceil - \lceil \{t_0\} \rceil) \in \{-\frac{1}{2},0,\frac{1}{2}\},$$
		where $\lceil \{t_i\} \rceil \in \{0,1\}$ is the ceiling function applied to $t_i$.
		Then, we set the transition guard of the corresponding edge of $\discta{A}$ to be
		$$t =  w(e) + \delta.$$
		In case of a weight $w(e) = m^*$ then the transition guard is
		$$t \geq m+\delta.$$
	\end{enumerate}
\end{definition}
A discretized timed automaton can be seen in Fig.~\ref{fig:disc_ta}(c).

We remark that the fact that the transition guards of $\discta{A}$ are over $\frac{1}{2} \ZNaturals$ and not over $\ZNaturals$ need not bother us since the standard definition of timed automata holds also over the rational numbers. Indeed, by letting all clocks run twice as fast and multiplying by 2 all values in the constraints of the transition guards, we end up in an automaton over the integers.
\begin{figure}[]
	\centering
	\scalebox{0.5}{ 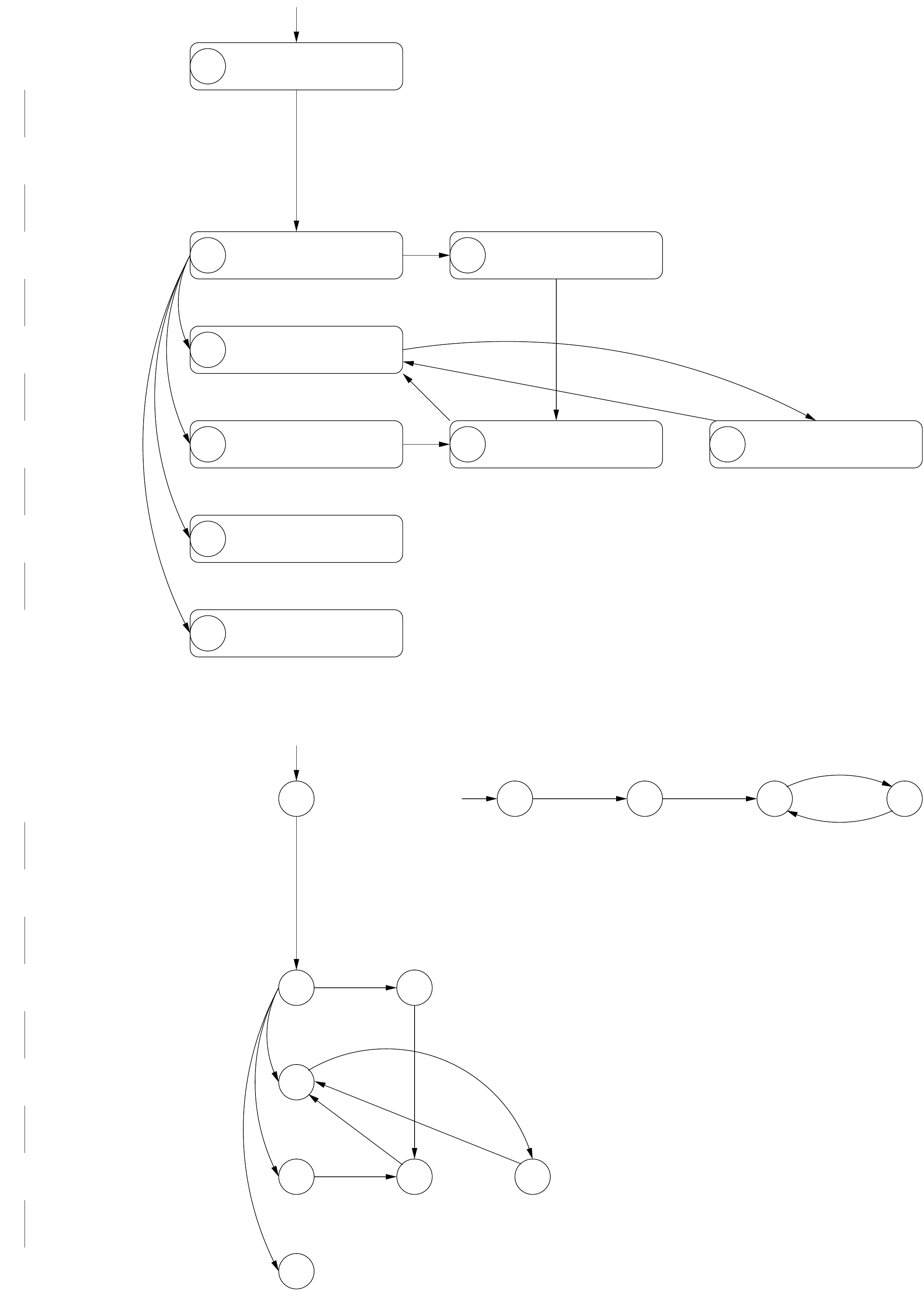 }
	\caption{a) $A \in \ta$; b) $\ARRR(A)$, the augmented region automaton of $A$; c) $\discta{A}$, the discretized timed automaton with $\Delta=0.5$.}
	\label{fig:disc_ta}
\end{figure}

\section{The Conformance Distance}
\label{sec:inc_gap}
We want to define a metric on the set of timed traces in order to define (conformance) distance between timed languages. 
\begin{definition}
	Given a set $T$ of timed traces over the same alphabet $\Sigma$, we define the $\infty$-metric or max-metric $d$ on $T$ in the following way. Given two timed traces
	\begin{align*}
		\tau_1 &=(t_1^{\tau_1}, a_1^{\tau_1}), (t_2^{\tau_1}, a_{2}^{\tau_1}),\ldots,(t_m^{\tau_1}, a_m^{\tau_1}), \\
		\tau_2 &=(t_1^{\tau_2}, a_1^{\tau_2}), (t_2^{\tau_2}, a_{2}^{\tau_2}),\ldots,(t_n^{\tau_2}, a_n^{\tau_2}),
	\end{align*}
the distance between $\tau_1$ and $\tau_2$ is
\begin{equation*}
	d(\tau_1, \tau_2) = \begin{cases}
	\infty, &\quad \mathrm{if} \; m \neq n \; \mathrm{or} \; a_i^{\tau_1} \neq a_i^{\tau_2} \; \mathrm{for} \; \mathrm{some} \; i,\\
	\max_i{| t_i^{\tau_1}-t_i^{\tau_2} |}, &\quad \mathrm{otherwise.}
	\end{cases}
\end{equation*}
\end{definition}
The above metric over the set of traces induces inclusion relation on timed languages (languages of timed automata).
\begin{definition}
	Given two timed languages $\LLL_1$ and $\LLL_2$, $\LLL_1$ is $\varepsilon$-included in $\LLL_2$, denoted $\LLL_1 \subseteq_{\varepsilon} \LLL_2$, if for every timed trace $\tau_1 \in \LLL_1$ there exists a timed trace $\tau_2 \in \LLL_2$ such that $d(\tau_1, \tau_2) \leq \varepsilon$. \\
	The \emph{conformance distance} $c(\LLL_1, \LLL_2)$ between $\LLL_1$ and $\LLL_2$ is
	$$c(\LLL_1, \LLL_2) = \inf\{\varepsilon \, : \,  \LLL_1 \subseteq_{\varepsilon} \LLL_2\},$$
	that is,
	$$c(\LLL_1, \LLL_2) = \sup_{\tau_1 \in \LLL_1} \inf_{\tau_2 \in \LLL_2} d(\tau_1, \tau_2) = \sup_{\tau_1 \in \LLL_1} d(\tau_1, \LLL_2).$$
	The distance $d(\LLL_1,\LLL_2)$ between $\LLL_1$ and $\LLL_2$ is 
	\begin{equation}
	\label{qe:lang_dist}
	d(\LLL_1,\LLL_2) = \max\{c(\LLL_1, \LLL_2), c(\LLL_2, \LLL_1)\}.
	\end{equation}
\end{definition}

In case of a finite conformance distance $n$ that is reached as a limit of a sequence of distances, we can denote it as $n^{+}$ (for a limit from above) or as $n^{-}$ (for a limit from below). Thus, $\LLL_1 \subseteq \LLL_2$ if and only if $c(\LLL_1,\LLL_2) = 0$.
But when $c(\LLL_1,\LLL_2) = 0^{+}$ then $\LLL_1 \nsubseteq \LLL_2$ but $\LLL_1 \subseteq \overline{\LLL_2}$, where $\overline{\LLL_2}$ is the closure of $\LLL_2$ in the Euclidean topology, defined as follows.
Fixing an untimed word $w \in \Sigma^*$ of length $n$, let $\LLL_2(w)$ be the timed traces in $\LLL_2$ whose untimed word is $w$ and let $\Reals ^n_w$ be a copy of $\Reals^n$ indexed by $w$.
There is a natural embedding $\iota: \LLL_2(w) \to \Reals ^n_w$.
Then, $c(\LLL_1, \LLL_2)=0^{+}$ implies that $\iota(\LLL_1) \subseteq \overline{\iota(\LLL_2)}$, where $\iota(\LLL_j) = \bigcup_{w \in \Sigma^*}\iota(\LLL_j(w))$, $j=1,2$, and  $\overline{S}$ is the closure of $S$ in the Euclidean topology. 
	
Subadditivity (triangle inequality) holds for the conformance distance:
$$c(\LLL_1, \LLL_3) \leq c(\LLL_1, \LLL_2) + c(L_2, L_3).$$
\begin{theorem}
	\label{th:integer_gap}
	Let $A,B \in \nta$. Then $c(\LLL(A), \LLL(B)) \in \frac{1}{2}\ZNaturals \cup \{\infty\}$.
\end{theorem}
\begin{proof}
	 Clearly, the conformance distance $c(\LLL(A), \LLL(B))$ can be $\infty$, for example, when the untimed language of $A$ contains a word that is not in the untimed language of $B$.
	 Suppose now that $\delta = c(\LLL(A), \LLL(B)) < \infty$.
	 It suffices to show the following. Given a path $\gamma^A$ in $A$ and another path $\gamma^B$ in $B$, where both define the same untimed trace (identical sequence of actions), let $T^A$ ($T^B$) be the set of all timed traces along $\gamma^A$ ($\gamma^B$).
	 We need to show that
	 \begin{equation}
	 \label{eq:dist_path}
	 \sup_{\tau^A \in T^A} \inf_{\tau^B \in T^B} d(\tau^A, \tau^B) \in \frac{1}{2}\ZNaturals.
	 \end{equation}
	By \cite{R19}, the timestamp of each of the events along $\gamma^A$ and $\gamma^B$ is an interval of the form $(m,n)$, $(m,n]$, $[m,n)$ or $[m,n]$, where $m \leq n$ and $m \in \ZNaturals$, $n \in \ZNaturals \cup \{\infty\}$.
	This can be shown by writing equalities and inequalities over the integers and variables $z_i$, where $z_i$ represents the time of the $i$-th event along the path.
	Then \eqref{eq:dist_path} becomes an optimization problem over the integers and variables for the events along $\gamma^A$ as well as for those along $\gamma^B$.
	The solution lies in $\frac{1}{2}\ZNaturals$ because it can be shown that for any other solution the timed traces can be shifted so that we are nearer $\frac{1}{2}\ZNaturals$.
	In fact, it is quite clear that the solution should be looked for when considering the integral end-points of the event intervals. The solution is, in general, in $\frac{1}{2}\ZNaturals$ and not in $\ZNaturals$ as can be seen from the following example. Suppose that an event of $\tau^A$ occurs at time $0 < t < 1$ where the corresponding event in $B$ can occur at time $0$ or at time $1$. Then, the maximal time difference, namely $\frac{1}{2}$, occurs when we choose the event of $\tau^A$ to be at time $t=\frac{1}{2}$.
\end{proof}
	
By the way they are defined, the untimed runs on the augmented region automaton $\ARRR(A)$, as well as those on the discretized timed automaton $\discta{A}$, are identical to the untimed runs on $A$.
The runs differ in the exact time on which each event occurs.
When the absolute time of occurrence of an event is $t_0 \in \ZNaturals$ then $\discta{A}$ agrees with $A$.
When $t_0 = n+\varepsilon$, $n \in \ZNaturals$, $0 < \varepsilon <1$ then the time of the event on $\discta{A}$ is set to be $n + \frac{1}{2}$, thus, the time difference is less than $\frac{1}{2}$ time units.
The fact that the clock $t$ of $\discta{A}$ is synchronized with the clock $t$ that was added to $A$ to measures absolute time guarantees that the cumulative error does not increase over time but remains bounded by $\frac{1}{2}$. That is, $\discta{A}$ is a $\frac{1}{2}$-time-unit approximation of $A$: there exits a surjective mapping
\begin{equation}
	\label{eq:proj}
	\pi : \LLL(A) \twoheadrightarrow \LLL(\discta{A}),
\end{equation}
such that if $\pi(\tau) = \tilde{\tau}$ then $d(\tau, \tilde{\tau}) < \frac{1}{2}$.
We showed that the following holds.
\begin{theorem}
	\label{th:approx_ta}
	$d(\LLL(A), \LLL(\discta{A})) \leq \frac{1}{2}$.
\end{theorem}
Since $t$ is reset only on values in $\frac{1}{2}\ZNaturals$ then $\discta{A}$ is determinizable (see \cite{irta}, \cite{bbbb}).
In fact, since $t$ is reset at each transition, we can remove it altogether to obtain an action-labeled, weighted directed graph.
The determinization algorithm is then straightforward by searching the graph in a breadth-first manner, unifying edges of the same source location that agree on their labels: $(a,t)$, $a$ - action, $t$ -time, followed by unifying the target locations.
The number of vertices, however, may grow exponentially.

\section{Computing the Conformance Distance}
Since $\discta{A}$ is determinizable, we can gain information about the relation between the languages of two timed automata by comparing their discretized languages.

Note that by the way the distance between languages is defined, it is clear that it refers to languages which are supposed to be (almost) identical or that one language is assumed to be (almost) included in the other, but this is normally the case in equivalence verification or when comparing the implication language with its specification.
Note that even if the untimed languages of two TAs are identical, it is enough that there exists a cycle, in which the timed languages do not agree, then by repeatedly taking this cycle the distance between the timed traces of the two TAs may grow indefinitely, resulting in a distance of $\infty$, and it is of interest to be able to recognize when this phenomenon occurs.
Thus, it seems that since the distance between $A$ and its discretized timed automaton $\discta{A}$ is only $\frac{1}{2}$ time units, we may not lose much by comparing $\discta{A}$ instead of $A$ with another TA.
In fact, in order to be more precise in the computation of the distance between two languages we need to make the basic discretization interval, denoted $\Delta$, shorter than $\frac{1}{2}$ time units.
By setting $\Delta = \frac{1}{n}$ we get that $d(\LLL(A), \LLL(\discta{A})) < \frac{1}{n}$, thus we can make $\LLL(A)$ and $\LLL(\discta{A})$ as close to one another as we like (of course, in the expense of complexity). 
However, it turns out that it suffices to choose $\Delta = \frac{1}{6}$ in order to get the maximal precision about $d(\LLL(A), \LLL(B))$.

For our convenience, since we prefer not to work with small fractions we accelerate the clocks to run at triple speed. That is, from now on, given the timed automata $A$ and $B$ under test, we first multiply by $3$ all the numbers that appear in the transition guards, so they all belong to $3\ZNaturals$.
Then we proceed as before: we construct the region automata with respect to basic regions of size 1 time unit and the discretized automata with respect to $\Delta = \frac{1}{2}$.
Now we have,
\begin{equation}
\label{eq:gap_set}
c(\LLL(A), \LLL(B)) \in \frac{3}{2}\ZNaturals \cup \{\infty\}
\end{equation}
and
$$
d(\LLL(A), \LLL(\discta{A})) \leq \frac{1}{2}, \quad d(\LLL(B), \LLL(\discta{B})) \leq \frac{1}{2}.
$$
\begin{theorem}
	\label{th:inclusion_gap}
	Let $A,B \in \nta$ with clocks running at triple speed and let $\discta{A}, \discta{B}$ be their discretized timed automata with respect to $\Delta = \frac{1}{2}$.
	Then
	$$
	|\, c(\LLL(A), \LLL(B)) - c(\LLL(\discta{A}), \LLL(\discta{B})) \, | \leq \frac{1}{2}
	$$
	and $c(\LLL(A), \LLL(B))$ is known in case $c(\LLL(\discta{A}), \LLL(\discta{B}))$ is known.
In particular:
$$
	\LLL(\discta{A}) \nsubseteq \LLL(\discta{B}) \; \Rightarrow \; \LLL(A) \nsubseteq \LLL(B)
$$
and
$$
	\LLL(\discta{A}) \subseteq \LLL(\discta{B}) \; \Rightarrow \; \LLL(A) \subseteq \overline{\LLL(B)},
$$		
so that the language inclusion problem between $\LLL(A)$ and the topological closure of $\LLL(B)$ is decidable. 
\end{theorem}
\begin{proof}
	$A$ and $\discta{A}$ have the same untimed language.
	The timed languages $\LLL(A)$ (with clocks running at triple speed) and $\LLL(\discta{A})$ differ from one another in that every event of a run on $A$ that occurs at time  $t$, with $t = n + \varepsilon$ and $0 < \varepsilon < 1$, occurs at the `rounded` time $n+\frac{1}{2}$ in the corresponding run on $\discta{A}$. Similarly for $B$ with respect to $\discta{B}$.
	It follows that $\delta = c(\LLL(A), \LLL(B)) = \infty$ if and only if $\delta_d = c(\LLL(\discta{A}), \LLL(\discta{B})) = \infty$.
	
	Suppose now that $\delta < \infty$.
	We know \eqref{eq:gap_set} that $\delta \in \frac{3}{2}\ZNaturals$.
	Since the timed traces of $\LLL(\discta{A})$ and $\LLL(\discta{B})$ are discretized to the set $\frac{1}{2}\ZNaturals$ then, when computing $\delta_d$ instead of $\delta$, we may have a difference of $\frac{1}{2}$ time units between the two.
	It follows that
	\begin{equation}
	\label{eq:delta}
	\delta = \begin{cases}
	3k, &\quad \mathrm{if} \; \delta_d \in \{3k-\frac{1}{2}, 3k, 3k+\frac{1}{2}\}, \\
	3k+\frac{3}{2}, &\quad \mathrm{if} \; \delta_d \in \{3k+1, 3k+\frac{3}{2}, 3k+2 \}.
	\end{cases}
	\end{equation}
	
	Let us elaborate on that.
	When $\delta$ is exactly $k$ and not $3k^{+}$ or $3k^{-}$ then it means that it is achieved on specific timed traces and not as a limit.
	That is, it refers to an even occurring at time $t^A$ on a run on $A$ and and event occurring at time $t^B$ on a run on $B$, with $| t^A - t^B | = 3k$.
	Since the fractional parts of $t^A$ and $t^B$ are identical, the discretization in the corresponding runs on $\discta{A}$ and $\discta{B}$ are identical so that they occur at times $t^{\discta{A}}$ and $t^{\discta{B}}$ with $| t^{\discta{A}} - t^{\discta{B}} | = 3k$.
	The same applies when $\delta$ is exactly $3k+\frac{3}{2}$ since we are working with a resolution of $\frac{1}{2}$.
	
	When $\delta = 3k^{+}$ or $\delta = 3k^{-}$ then it is achieved as a limit of timed traces.
	If $\delta = 3k^{+}$ then $\delta_d$ can be $3k + \frac{1}{2}$, for example, when $t^A=3$ and $t^B = 3 + \varepsilon$, $\varepsilon>0$.
	Then the discretized traces will occur at times $t^{\discta{A}}=3$ and $t^{\discta{B}}=3\frac{1}{2}$. Then by choosing a sequence of timed traces of $\LLL(B)$ the time difference can tend to $0$ while in the discretized automata it will remain $\frac{1}{2}$.
	
	The other cases of an conformance distance $\delta$ that is a limit of converging distances are analogous, but we do not go here into detail.
	
	Let us look at the last claims of the theorem.
	Suppose that $\LLL(A) \subseteq \LLL(B)$.
	Then for each timed trace of $\LLL(A)$ there is an identical timed trace of $\LLL(B)$.
	The projection to the discretized timed trace will also be identical, thus, 
	$$
	\LLL(A) \subseteq \LLL(B) \; \Rightarrow \; \LLL(\discta{A}) \subseteq \LLL(\discta{B}).
	$$
	If $\LLL(A) \nsubseteq \overline{\LLL(B)}$ then $\delta > 0$.
	By \eqref{eq:delta}, we have that $\delta_d > 0$.
	It follows that  
	$$
	 \LLL(A) \nsubseteq \overline{\LLL(B)} \; \Rightarrow \; \LLL(\discta{A}) \nsubseteq \LLL(\discta{B}).
	$$	
\end{proof}

By $\eqref{qe:lang_dist}$, a similar result to Theorem~\ref{th:inclusion_gap} holds with respect to distances between languages.

By Theorem~\ref{th:inclusion_gap}, in order to compute the conformance distance $c(\LLL(A), \LLL(B))$, we can compute $c(\LLL(\discta{A}), \LLL(\discta{B}))$, and know that we lie within an error of at most $\frac{1}{2}$ time unit.
We may assume that $\discta{A}$ 
is deterministic, as this is feasible.
It is not necessary to determinize $\discta{B}$.

The general goal in computing $c(\LLL(\discta{A}), \LLL(\discta{B}))$ is to find the timed trace of $\LLL(\discta{A})$ that is farthest from $\LLL(\discta{B})$ (or a sequence of such timed traces if the distance is $\infty$).
A heuristic approach is to play a timed game in which the player in white moves along $\discta{A}$ and tries to maximize her wins, while the player in black moves along $\discta{B}$ and tries to minimize his losses.
The players start from the initial vertex of each graph.
Then white makes a move by jumping to a vertex in $\discta{A}$ with transition label $a$, followed by a move of black on an edge in $\discta{B}$ with the same label $a$.
Next, white moves on an edge with label $a'$, followed by a move of black with the same label $a'$ and so on. At each move we record the time difference between the absolute time duration of the paths along $\discta{A}$ and along $\discta{B}$. 
The problem is that we may return to the same pair of locations $(q,q') \in \Locs^{\discta{A}} \times \Locs^{\discta{B}}$ but with a different time difference between the path along $\discta{A}$ and that along $\discta{B}$. In addition, there are moves to locations where the time is not a single value but of the form $t \geq m$.
Thus, the game may not be of finite type.
One strategy to cope with the complexity of the game is a greedy max-min algorithm: each move of white is one that maximizes the new difference in times after the following move of black that tries to minimizes the time difference.
A better, but more expensive, strategy on the part of white is to look-ahead more than one step.

So, let us then consider a seemingly easier question: is $c(\LLL(\discta{A}), \LLL(\discta{B}))$ finite or infinite?
For this question we do not need to speed-up the clocks.
An infinite conformal distance occurs in one of the following three situations.
\begin{enumerate}
	\item[\textbf{S1.}] The untimed language of $\discta{A}$ is not included in that of $\discta{B}$: there exists a path $q_0 \xrightarrow{a_1} q_1 \xrightarrow{a_2} \cdots \xrightarrow{a_n} q_n$ in $\discta{A}$, with $q_n$ an accepting location, which either cannot be realized in $\discta{B}$ with the same sequence of actions, or all such paths in $\discta{B}$ do not terminate in an accepting location.
	\item[\textbf{S2.}] There exists a path in $\discta{A}$ of the form $q_0 \xrightarrow{a_1} q_1 \xrightarrow{a_2} \cdots \xrightarrow{a_n} q_n$, where the transition $q_{n-1} \xrightarrow{a_n} q_n$ has guard $t \geq m$, whereas for any path in $\discta{B}$ of the form $q'_0 \xrightarrow{a_1} q'_1 \xrightarrow{a_2} \cdots \xrightarrow{a_n} q'_n$ the guard of the last transition $q'_{n-1} \xrightarrow{a_n} q'_n$ bounds $t$ from above.
	\item[\textbf{S3.}] For each $N \in \Naturals$ there exists a timed trace $\tau \in \LLL(\discta{A})$, such that for each $\sigma \in \LLL(\discta{B})$, $d(\tau, \sigma) > N$ and not because of S2.
\end{enumerate}  
In order to find out whether the conformance distance between $\LLL(\discta{A})$ and $\LLL(\discta{B})$ is infinite as a result of \textbf{S1} or \textbf{S2} we extend $\discta{A}$ and $\discta{B}$ as follows.

First, we add to the set $\Actions$ of actions a copy of it $\bar{\Actions} = \{ \bar{a} \, : \, a \in \Actions \}$.
Then, for each transition $q \xrightarrow{a} q'$ of $\discta{A}$ or of $\discta{B}$ with time constraint of type $t \geq m$, we add a transition $q \xrightarrow{\bar{a}} q'$ with guard $t = \infty$.
Next, we complete $\discta{B}$ by adding a location $s$ which is a 'sink': whenever there is no transition with action $b \in \Actions \cup \bar{\Actions}$ from location $q$ of $\discta{B}$, we add the transition $q \xrightarrow{b} s$.
The sink location is supplemented by self-loops of all actions.
We retain the names $\discta{A}$ and $\discta{B}$ for the resulting automata. 

In the next step we form the untimed automaton $U(\discta{A})$ which is a determinization of  $\discta{A}$ with respect to actions while ignoring the temporal part.
Similarly, we construct $U(\discta{B})$.
\begin{definition}
	\label{def:uappta_a}
	The automaton $U(\discta{A})$ is a tuple $(\Locs, Q_0, \Acc, \Actions \cup \bar{\Actions}, \Edges)$, where:
	\begin{enumerate}
		\item $\Locs \subseteq \mathcal{P}(\Locs^{\discta{A}})$ is a subset of the power set of the locations of $\discta{A}$, where $Q_0 = \{q^{\discta{A}}_0\}$ is the initial location; 
		\item $\Acc \subseteq \Locs$ is the set of accepting locations, where $Q = \{q^{\discta{A}}_1, \ldots, q^{\discta{A}}_m\}$ is accepting if at least one of the $q^{\discta{A}}_i$ is an accepting location of $\discta{A}$;
		\item $\Actions \cup \bar{\Actions}$ is the set of actions;
		\item $\Edges \subseteq \Locs \times (\Actions \cup \bar{\Actions}) \times \Locs$ is a finite set of edges of the form $(Q, a, Q')$, where
		$Q' = \{ q'^{\discta{A}} : \exists q^{\discta{A}} \in Q. \, (q^{\discta{A}}, a, q'^{\discta{A}} )\in \Trans^{\discta{A}} \}$.
	\end{enumerate}
\end{definition}
Finally, we construct a version of the untimed product automaton $U(\discta{A}) \times U(\discta{B})$ in which the accepting locations are those pairs of locations $(Q,Q')$ for which $Q$ is an accepting location of $U(\discta{A})$ but $Q'$ is not an accepting location of $U(\discta{B})$.
\begin{definition}
	\label{def:prod}
	The product automaton $U(\discta{A}) \times U(\discta{B})$ is a tuple $(\Locs, Q_0, \Acc, {\Actions \cup \bar{\Actions}}, \Edges)$, where:
	\begin{enumerate}
		\item $\Locs \subseteq \Locs^{U(\discta{A})} \times \Locs^{U(\discta{B})}$, where $Q_0 = (q^{U(\discta{A})}_0, q^{U(\discta{B})}_0)$ is the initial location; 
		\item $\Acc \subseteq \Locs$ is the set of accepting locations $(Q,Q')$, where $Q \in \Acc^{U(\discta{A})}$ and $Q' \notin \Acc^{U(\discta{B})}$;
		\item $\Actions \cup \bar{\Actions}$ is the set of actions;
		\item $\Edges \subseteq \Locs \times (\Actions \cup \bar{\Actions}) \times \Locs$ is the set of edges, where for each $(Q_1, a, Q'_1) \in \Edges^{U(\discta{A})}$ and $(Q_2, a, Q'_2) \in \Edges^{U(\discta{B})}$ we have that $((Q_1,Q_2), a, (Q'_1,Q'_2)) \in \Edges$,
	\end{enumerate}
	and $U(\discta{A}) \times U(\discta{B})$ is the connected component of the initial location.
\end{definition}
\begin{theorem}
	$c(\LLL(\discta{A}), \LLL(\discta{B})) = \infty$ as a result of \textbf{S1} or \textbf{S2} if and only if the set of accepting locations of $U(\discta{A}) \times U(\discta{B})$ is not empty.
\end{theorem}  
\begin{proof}
	By completing $\discta{B}$ we made sure that the set of untimed traces of $U(\discta{B})$ consists of all possible traces.
	But if a path in $U(\discta{A}) \times U(\discta{B})$ terminates in an accepting location, then there exists a path in $\discta{A}$ that ends in an accepting location, while all paths in $\discta{B}$ of the same sequence of actions either terminate in a non-sink location which is non-accepting, or enter the sink either on an edge with action $a \in \Actions$ due to missing such an edge on the uncompleted $\discta{B}$ or on an edge labeled $\bar{a} \in \bar{\Actions}$ due to reaching a transition that is bounded in time in (the uncompleted) $\discta{B}$, but not bounded in $\discta{A}$.
\end{proof}

Assume now that by constructing the automaton $U(\discta{A}) \times U(\discta{B})$ it turns out that no possible infinite conformance distance exists when checking \textbf{S1} and \textbf{S2} and it remains to check \textbf{S3}.
Hence, the goal is to find a sequence of traces in $\discta{A}$ which 'run-away' from $\discta{B}$, and now we are interested in the exact delays between consecutive transitions.
This problem may be of very high complexity and even it is not clear whether it is decidable.
We will show that a (perhaps) restricted version is decidable.

First, we extend $\discta{A}$ and $\discta{B}$ with actions $\bar{\Actions}$ as before, referring to transitions that are unbounded by time.
Let $M$ be the maximal integer that appears in a transition guard of $\discta{A}$ or $\discta{B}$.
Then, each transition $q \xrightarrow{a} q'$ of $\discta{A}$ or $\discta{B}$ with time constraint $t \geq m$, $m \leq M+\frac{1}{2}$, is replaced by the transitions $q \xrightarrow{a} q'$ with delays $t = m$, $t=m+\frac{1}{2}$,..., $t=M+\frac{1}{2}$ and another transition $q \xrightarrow{\bar{a}} q'$ with delay $t = (M+1)^*$. The set of delays of $\discta{A}$ ($\discta{B}$) is denoted by $\Delays$.

In the next step we determinize $\discta{B}$ into $D(\discta{B})$.
The idea is to be able to compare each timed trace of $\discta{A}$ simultaneously with all equivalent (having the same untimed trace) time traces of $\discta{B}$.
\begin{definition}
	\label{def:detappta_b}
	The automaton $D(\discta{B})$ is a tuple $(\Locs, Q_0, \Acc, \Actions \cup \bar{\Actions}, \Trans)$, where:
	\begin{enumerate}
		\item $\Locs \subseteq \mathcal{P}(\Locs^{\discta{B}})$ is a subset of the power set of the locations of $\discta{B}$, where $Q_0 = \{q^{\discta{B}}_0\}$ is the initial location; 
		\item $\Acc \subseteq \Locs$ is the set of accepting locations, where location $Q = \{q^{\discta{B}}_1, \ldots, q^{\discta{B}}_m\}$ is accepting if at least one of the $q^{\discta{B}}_i$ is an accepting location of $\discta{B}$;
		\item $\Actions \cup \bar{\Actions}$ is the set of actions;
		\item $\Trans \subseteq \Locs \times (\Actions \cup \bar{\Actions}) \times \Edges \times \Locs$ is a finite set of transitions of the form $(Q, a, E, Q')$, where
		$$E = \{(q^{\discta{B}}, a, d, q'^{\discta{B}}) \in \Trans^{\discta{B}} \, : \,  q^{\discta{B}} \in Q, q'^{\discta{B}} \in Q', a \in \Actions \cup \bar{\Actions}, d \in \Delays	\}.$$
		and $Q'$ contains exactly the set of these target locations $q'^{\discta{A}}$.
	\end{enumerate}
\end{definition} 
	Note that the transitions of $D(\discta{B})$ retain the set of transitions of $\discta{B}$ including source and target locations.

	In the next step we make the standard construction of the product automaton $\discta{A} \times D(\discta{B})$. It has at most $L = | \Locs ^{\discta{A}} | \cdot 2^{|\Locs^{\discta{B}}|}$ locations, where each location is of the form
	$$
	 Q^{\discta{A} \times D(\discta{B})} = (q^{\discta{A}}, \{q^{\discta{B}}_1, \ldots, q^{\discta{B}}_m\}).
	$$
	Since the difference between a transition delay on $\discta{A}$ and a corresponding transition on $\discta{B}$ in parallel runs on $\discta{A}$ and $\discta{B}$ is at most $M$ time units (actually, it is $M+\frac{1}{2}$, but it makes no difference for our argument), then a run on $\discta{A} \times D(\discta{B})$ that does not visit the same location twice may result in a delay of at most $LM$ time units between its projection to $\discta{A}$ and each of its projections to $\discta{B}$.
	
	At each transition of a run on $\discta{A} \times D(\discta{B})$ we can subtract the delay of the edge of $\discta{A}$ from each of the delays of the corresponding edges of $\discta{B}$ and record at each location $q^{\discta{B}}_i$ of $Q^{\discta{A} \times D(\discta{B})}$ the set of \emph{accumulated time differences} (ATDs), that is, the differences in absolute time between the runs on $\discta{A}$ and all possible runs on $\discta{B}$ of the same untimed trace when reaching $Q^{\discta{A} \times D(\discta{B})}$.
	The ATD of the least absolute value gives the least difference in time at that location between the run on $\discta{A}$ and a corresponding run on $\discta{B}$.
	When a delay is $(M+1)^*$ (and then it is the same delay for both $\discta{A}$ and $\discta{B}$) then we denote the difference $0^+$, and this $+$ sign carries on to the next differences by defining $i^+ + j = (i+j)^+$, $i^+ + j^+ = (i+j)^{++}$, and so on. It means that $i^+$ is actually any value of $\frac{1}{2}\ZNaturals$ which is greater than or equals $i$.
	The reason for that is that for a delay $k$ in $\discta{A}$ we can choose any delay $l \geq k$ in $\discta{B}$. In order to exclude the possibility of choosing also a delay in $\discta{B}$ which is smaller than $k$ (and maybe reduce the distance between the corresponding paths in $\discta{A}$ and $\discta{B}$), each transition of $\discta{A}$ that is unbounded in time is considered as a delay of $M+1$ time units.
	Once a value of the form $i^+$ is realized as a concrete value $i+j$, for some $j \geq 0$, then in all the difference values that appear in the following locations the relevant $+$ sign is removed and the value $j$ is added.
	
	Every run $\rho$ on $\discta{A} \times D(\discta{B})$ can be uniquely written in the form
	$$
	\rho = \rho_0 \sigma_1^{i_1} \rho_1 \sigma_2^{i_2} \rho_2 \cdots \sigma_r^{i_r} \rho_r,
	$$
	for some $r \in \ZNaturals$, $i_j \in \Naturals$ and where each $\sigma_j$ is a simple cycle of positive length
	and each $\rho_j$ is without cycles and of length $0 \leq l < L$.
	We say that the number of \emph{power cycles} of $\rho$ is $r$, written $pc(\rho)=r$.
	\begin{theorem}
		It is decidable whether there exists a fixed $K \in \Naturals$, such that for every $N \in \Naturals$ there exists a timed trace $\tau \in \LLL(\discta{A})$, such that $d(\tau, \LLL(\discta{B})) > N$ and the corresponding run $\rho$ on $\discta{A} \times D(\discta{B})$ satisfies $pc(\rho) \leq K$. 
	\end{theorem}  
	\begin{proof}
		The conformance distance $c(\LLL(\discta{A}), \LLL(\discta{B}))$ is $\infty$ if for every $N \in \Naturals$ we can reach a location $Q^{\discta{A} \times D(\discta{B})}$ with all ATDs of absolute value at least $N$.
		Since $pc(\rho) \leq K$, $K$ fixed, it is clear that the unbounded increase in the ATDs can come only from the powers of simple cycles $\sigma^{i_j}$.
		Since the number of locations of $\discta{A} \times D(\discta{B})$ is finite, all locations can be reached in a bounded number of steps.
		Then, for each location $Q^{\discta{A} \times D(\discta{B})}$ and each simple cycle $\sigma$ starting at $Q^{\discta{A} \times D(\discta{B})}$, it can be checked for which locations $q^{\discta{B}}_i$ of $Q^{\discta{A} \times D(\discta{B})}$ the minimal (in absolute value) ATD increases indefinitely when repeating the cycle $\sigma$.
		Let $P$ be the set of these locations $Q^{\discta{A} \times D(\discta{B})}$ with at least one unbounded ATD.
		
		Next we look at all the simple paths from each $Q^{\discta{A} \times D(\discta{B})} \in P$ to the other locations of $\discta{A} \times D(\discta{B})$ and update their sub-locations $q^{\discta{B}}_i$ of having an unbounded ATD.
		Moreover, when reaching a location $Q'^{\discta{A} \times D(\discta{B})} \in P$ from another location $Q^{\discta{A} \times D(\discta{B})} \in P$ then it can be checked whether new sub-locations $q^{\discta{B}}_i$ of $Q'^{\discta{A} \times D(\discta{B})}$ become with unbounded ATD when repeating a cycle $\sigma$ (even when its minimal ATD decreases by a bounded finite number at each round of $\sigma$, if we started with an unbounded value, we can end at an unbounded value). 
		This process is repeated until no improvement in the maximum number of sub-locations of unbounded ATD can be achieved.
		Since the graph is finite, the whole algorithm is finite.
		Finally, $c(\LLL(\discta{A}), \LLL(\discta{B})) = \infty$ when at some step of the algorithm a location $Q^{\discta{A} \times D(\discta{B})}$ becomes with all its sub-locations $q^{\discta{B}}_i$ of unbounded ATD.
	\end{proof}
\section{Conclusion and Suggested Future Research}
In this paper we introduced a natural definition of the distance between the languages of non-deterministic timed automata in terms of the times at which events in one automaton occur compared to the times of corresponding events in the other automaton. We showed how one can effectively construct discretized deterministic timed automata and obtain the distance between the original timed automata from the distance between the discretized versions. Consequently, the problem of language inclusion for timed automata, which is undecidable in general, is decidable if we consider the closure of the languages with respect to the Euclidean topology.

Computing the distance between timed automata may not be an easy task. We even do not know whether the finiteness of the distance is a decidable problem.
We showed, however, that under some restriction on the timed traces, this problem is decidable.

There is more than one reasonable way to define the distance between timed automata and the one we chose refers to the accumulated time difference that may occur between timed automata that are supposed to be (almost) the same or conformance distance between the language of an implementation and that of the specification.
Other possible definitions of distances like a maximal time difference on a single transition or time difference mean on simple cycles are easier to compute on the discretized automata.
For another notion of distance between implementation and specification we refer to \cite{CHR12}.    
Another interesting problem is to compute the distance between timed automata equipped with probabilities on transitions, where the distances are computed as expected values with respect to these probabilities.

\subsection*{Acknowledgements.} 
\begin{small}
	This research was partly supported by the Austrian Science Fund (FWF) Project P29355-N35. 
\end{small}

\bibliography{ta}
\bibliographystyle{amsplain}

\end{document}